\documentclass[a4paper,UKenglish,cleveref, autoref, thm-restate,authorcolumns]{lipics-v2019}
\usepackage[T1]{fontenc}

\graphicspath{{./figures/}}%helpful if your graphic files are in another directory

\title{Safety in $s$-$t$ Paths, Trails and Walks} %TODO Please add

\author{Massimo Cairo}{Department of Computer Science, University of Helsinki, Finland}{}{}{}%TODO mandatory, please use full name; only 1 author per \author macro; first two parameters are mandatory, other parameters can be empty. Please provide at least the name of the affiliation and the country. The full address is optional

\author{Shahbaz Khan}{Department of Computer Science, University of Helsinki, Finland}{shahbaz.khan@helsinki.fi}{https://orcid.org/0000-0001-9352-0088}{}

\author{Romeo Rizzi}{Department of Computer Science, University of Verona, Italy}{romeo.rizzi@univr.it}{https://orcid.org/0000-0002-2387-0952}{}

\author{Sebastian Schmidt}{Department of Computer Science, University of Helsinki, Finland}{sebastian.schmidt@helsinki.fi}{https://orcid.org/0000-0003-4878-2809
}{}

\author{Alexandru~I.~Tomescu}{Department of Computer Science, University of Helsinki, Finland}{alexandru.tomescu@helsinki.fi}{https://orcid.org/0000-0002-5747-8350}{}

\authorrunning{M. Cairo, S. Khan, R. Rizzi, S. Schmidt and A. Tomescu} %TODO mandatory. First: Use abbreviated first/middle names. Second (only in severe cases): Use first author plus 'et al.'

\Copyright{Massimo Cairo, Shahbaz Khan, Romeo Rizzi, Sebastian Schmidt and Alexandru~I.~Tomescu} %TODO mandatory, please use full first names. LIPIcs license is "CC-BY";  http://creativecommons.org/licenses/by/3.0/

\ccsdesc[500]{Mathematics of computing~Paths and connectivity problems}
\ccsdesc[500]{Theory of computation~Graph algorithms analysis}
\keywords{Directed graph, connectivity problem, graph algorithm, strong bridge, strong articulation point, safety} %TODO mandatory; please add comma-separated list of keywords

\category{} %optional, e.g. invited paper

\relatedversion{} %optional, e.g. full version hosted on arXiv, HAL, or other respository/website
\supplement{}%optional, e.g. related research data, source code, ... hosted on a repository like zenodo, figshare, GitHub, ...

\acknowledgements{We thank Elia~C.~Zirondelli for useful discussions. This work was partially funded by the European Research Council (ERC) under the European Union's Horizon 2020 research and innovation programme (grant agreement No.~851093, SAFEBIO) and by the Academy of Finland (grants No.~322595, 328877).}%optional

\nolinenumbers %uncomment to disable line numbering

\hideLIPIcs  %uncomment to remove references to LIPIcs series (logo, DOI, ...), e.g. when preparing a pre-final version to be uploaded to arXiv or another public repository

\EventEditors{John Q. Open and Joan R. Access}
\EventNoEds{2}
\EventLongTitle{42nd Conference on Very Important Topics (CVIT 2016)}
\EventShortTitle{CVIT 2016}
\EventAcronym{CVIT}
\EventYear{2016}
\EventDate{December 24--27, 2016}
\EventLocation{Little Whinging, United Kingdom}
\EventLogo{}
\SeriesVolume{42}
\ArticleNo{XX}
\usepackage[linesnumbered,ruled,vlined]{algorithm2e}
\SetKw{Continue}{continue}

\usepackage{amsmath}
\usepackage{amssymb}
\usepackage{amsthm}

\usepackage{hyperref}
\Crefname{subsection}{Subsection}{Subsections}
\crefname{subsection}{subsection}{subsections}

\usepackage[dvipsnames]{xcolor}
\usepackage[utf8]{inputenc}
\usepackage{float}

\newtheorem{problem}[theorem]{Problem}

\newcommand{\detour}{\textsc{Detour}}

\newcommand{\tail}{\textsc{tail}}
\newcommand{\head}{\textsc{head}}

\newcommand{\EPmodel}{\textsc{stPaths}}
\newcommand{\ETmodel}{\textsc{stTrails}}
\newcommand{\EWmodel}{\textsc{stWalks}}
\newcommand{\VPmodel}{\textsc{$V$-stPaths}}
\newcommand{\VTmodel}{\textsc{$V$-stTrails}}
\newcommand{\multiVTmodel}{\textsc{$V$-stMTrails}}
\newcommand{\VWmodel}{\textsc{$V$-stWalks}}
\newcommand{\XPmodel}{\textsc{$X$-stPaths}}
\newcommand{\XTmodel}{\textsc{$X$-stTrails}}
\newcommand{\XWmodel}{\textsc{$X$-stWalks}}

\newcommand{\EP}{\textsc{MaxSafe~\EPmodel{}}}
\newcommand{\ET}{\textsc{MaxSafe~\ETmodel{}}}
\newcommand{\EW}{\textsc{MaxSafe~\EWmodel{}}}

\newcommand{\XP}{\textsc{MaxSafe~\XPmodel{}}}
\newcommand{\XT}{\textsc{MaxSafe~\XTmodel{}}}
\newcommand{\XW}{\textsc{MaxSafe~\XWmodel{}}}

\newcommand{\bridge}{$s$-$t$ bridge}
\newcommand{\bridges}{$s$-$t$ bridges}
\newcommand{\articulationpoint}{$s$-$t$ articulation point}
\newcommand{\articulationpoints}{$s$-$t$ articulation points}

\newcommand{\bridgesequence}{bridge sequence}

\newcommand{\xbridgesequence}{$X$-bridge sequence}
\newcommand{\articulationsequence}{articulation sequence}

\newcommand{\island}{bridge component}
\newcommand{\islands}{bridge components}

\newcommand{\articulationislands}{articulation components}

\newcommand{\forbiddenwalk}{walk breaker}
\newcommand{\forbiddenwalks}{walk breakers}

\newcommand{\Forbiddenwalks}{Walk breakers}
\newcommand{\ForbiddenWalk}{Walk Breaker}
\newcommand{\ForbiddenWalks}{Walk Breakers}
\newcommand{\fforbiddenwalk}{$F$-walk breaker}

\newcommand{\xforbiddenwalk}{$X$-walk breaker}
\newcommand{\xforbiddenwalks}{$X$-walk breakers}

\newcommand{\forbiddentrail}{trail breaker}
\newcommand{\forbiddentrails}{trail breakers}

\newcommand{\Forbiddentrails}{Trail breakers}
\newcommand{\ForbiddenTrail}{Trail Breaker}

\newcommand{\forbiddenstructure}{breaking structure}

\begin{document}
	
	\maketitle
	
	\begin{abstract}
	Given a directed graph $G$ and a pair of nodes $s$ and $t$, an \emph{$s$-$t$ bridge} of $G$ is an edge whose removal breaks all $s$-$t$ paths of $G$ (and thus appears in all $s$-$t$ paths). Computing all $s$-$t$ bridges of $G$ is a basic graph problem, solvable in linear time.
	
	In this paper, we consider a natural generalisation of this problem, with the notion of ``safety'' from bioinformatics. We say that a walk $W$ is \emph{safe} with respect to a set $\mathcal{W}$ of $s$-$t$ walks, if $W$ is a subwalk of all walks in $\mathcal{W}$. We start by considering the maximal safe walks when $\mathcal{W}$ consists of: all $s$-$t$ paths, all $s$-$t$ trails, or all $s$-$t$ walks of $G$. We show that the first two problems are immediate linear-time generalisations of finding all $s$-$t$ bridges, while the third problem is more involved. In particular, we show that there exists a compact representation computable in linear time, that allows outputting all maximal safe walks in time linear in their length.
	
	We further generalise these problems, by assuming that safety is defined only with respect to a subset of \emph{visible} edges. Here we prove a dichotomy between the $s$-$t$ paths and $s$-$t$ trails cases, and the $s$-$t$ walks case: the former two are NP-hard, while the latter is solvable with the same complexity as when all edges are visible. We also show that the same complexity results hold for the analogous generalisations of \emph{$s$-$t$ articulation points} (nodes appearing in all $s$-$t$ paths). 
	
	We thus obtain the best possible results for natural ``safety''-generalisations of these two fundamental graph problems. Moreover, our algorithms are simple and do not employ any complex data structures, making them ideal for use in practice.
	\end{abstract}
	
	\newpage
	\section{Introduction}
	
	Connectivity and reachability are fundamental graph-theoretical problems studied extensively in the literature \cite{diestel10,HandbookGT,cormen01introduction,skiena}. 
	A key notion underlying such algorithms is that of edges (or nodes) critical for connectivity or reachability. The most basic variant of these are bridges (or articulation points), which are defined as follows. A \emph{bridge} of an undirected graph, also referred as \emph{cut edge}, is an edge whose removal increases the number of connected components. Similarly, a \emph{strong bridge} in a (directed) graph is an edge whose removal increases the number of strongly connected components of the graph. (Strong) articulation points are defined in an analogous manner by replacing edge with node.
	
	Special applications consider the notion of bridge to be parameterised by the nodes that become disconnected upon its removal~\cite{ItalianoLS12,Tarjan76}. Given a node $s$, we say that an edge is an \emph{$s$~bridge} (also referred as \emph{edge dominators} from source $s$~\cite{ItalianoLS12}) if there exists a node $t$ that is no longer reachable from $s$ when the edge is removed. Moreover, given both nodes $s$ and $t$, an \emph{$s$-$t$ bridge} is an edge whose removal makes $t$ no longer reachable from $s$. 
	
	From this point onward we assume a fixed (directed) graph $G$ without multiedges, with $n$ nodes and $m$ edges, and two given nodes $s$ and $t$ of $G$. Since $s$-$t$ bridges are exactly the edges (i.e., the paths of length one) appearing in all $s$-$t$ paths, it is natural to generalise this notion by considering the \emph{paths} (i.e., of length two or more) appearing in all $s$-$t$ paths. An equivalent way of defining this problem is through the notion of \emph{safety}~\cite{TomescuMedvedev,Tomescu2017}. Given a set of walks $\mathcal{W}$, we say that a walk $W$ is \emph{safe} with respect to $\mathcal{W}$ if $W$ is a subwalk of all walks in $\mathcal{W}$. Our problem is obtained by taking $\mathcal{W}$ to be the set of all $s$-$t$ paths.\footnote{We will focus on \emph{maximal} safe walks, namely those that cannot be extended left or right without losing safety.} We will also consider other natural generalisations for $\mathcal{W}$, e.g. all $s$-$t$ trails and all $s$-$t$ walks, as we will discuss in \Cref{sec:contribution}.
	
	\subparagraph*{Motivation.} Safety is motivated by real-world problems whose computational formulation admits multiple solutions. For this reason, we will also refer to the set $\mathcal{W}$ as the \emph{candidate set}. By looking at the parts common to all solutions---the safe parts---one can make more informed guesses on what can be correctly reported from the data. This approach is more feasible than e.g.,~the common approach of simply enumerating all solutions.
	
	A notable example is the genome assembly problem from Bioinformatics: one is given a set of short genomic fragments (the \emph{reads}) and one needs to reconstruct the genome from which these were sequenced (see e.g.~\cite{DBLP:books/cu/MBCT2015} for more details).
	A common approach is to build a graph from the reads, called the \emph{genome graph}, and then to define a genome assembly solution as a certain type of walk in that graph.
	Practical assemblers do not assemble full genomes, because the genome graph may admit a large number of solutions, but instead efficiently output only shorter strings that are guaranteed to appear in the genome.
	
	A natural notion of a genome assembly solution is that of a circular walk in the genome graph covering every edge or node at least once~\cite{TomescuMedvedev,Tomescu2017}.
	Finding all maximal safe walks for the edge-centric solution set can be solved optimally: an optimal quadratic-time algorithm was given in \cite{DBLP:journals/talg/CairoMART19}, and an optimal output-sensitive algorithm was given in~\cite{Cairo2020pre}.
	In this paper we drop both the circularity and the covering requirements from this solution set.
	This yields more basic graph problems with more fundamental solution sets that can potentially be computed more efficiently in practice.  
	In addition, we introduce a novel generalisation of the problems that makes a subset of nodes and edges \emph{invisible} in the solution set. This models scenarios where we want to ignore some uncertain or complex parts of the graph, but still report if the walks flanking this region always appear as consecutive in any solution. Moreover, keeping them in the graph still allows them to impact the safety of other walks.
	As such, our problems also have potential applications to the practical genome assembly problem (see also \Cref{s:conclusions}).
	
	\subparagraph*{Related work.} Safety has several precursors, the closest being \emph{persistence}: an individual node or edge is called \emph{persistent} if it appears in all solutions to a problem on the given graph. Persistent nodes and edges have been studied for maximum independent sets~\cite{doi:10.1137/0603052}, maximum bipartite matchings~\cite{Costa1994143}, assignments and transportations~\cite{DBLP:journals/mmor/Cechlarova98}. Other previous notions include \emph{$d$-transversals}~\cite{DBLP:journals/jco/CostaWP11} (sets of nodes or edges intersecting every solution to the problem in at least $d$ elements), \emph{$d$-blockers}~\cite{DBLP:journals/dm/ZenklusenRPWCB09} (sets of nodes or edges whose removal deteriorates the optimum solution to the problem by at least $d$), or \emph{most vital nodes or edges}~\cite{DBLP:journals/networks/BazganFNNS19}.
	%\hint{check update. previous version seemingly suggested d blocker is same as most vital node or edge}

	%\todo{Shahbaz: s-t bridge algorithms}
	For undirected graphs, the classical algorithm by Tarjan~\cite{Tarjan74} computes all bridges and articulation points in linear time. However, for directed graphs only recently Italiano et al.~\cite{ItalianoLS12} presented an algorithm to compute all strong bridges and strong articulation points in linear time. They also showed that classical algorithms~\cite{Tarjan76,GabowT85} compute $s$ bridges in linear time. The $s$ articulation points (or {\em dominators}) are extensively studied resulting in several linear-time algorithms~\cite{AlstrupHLT99,BuchsbaumKRW05,BuchsbaumGKRTW08}.
	The $s$-$t$ bridges were studied as minimum $s$-$t$ cuts in network flow graphs, where an $s$-$t$ bridge is a cut of {\em unit} size.
	These cuts can be discovered iteratively in the residual graph of the classic Ford Fulkerson algorithm~\cite{FordF56} after pushing unit flow into the network.
	Contracting the first cut to $s$, the next $s$-$t$ bridge can be discovered, and so on.
	%Since flow is restricted to unit size,
	Since only unit sized flows are of interest, 
	the algorithm completes in linear time. Recently, this algorithm was simplified for unit sized cuts ($s$-$t$ bridges) by Cairo et al.~\cite{CairoKRSTZ20}. 
	
	%The classic Ford Fulkerson algorithm~\cite{FordF56} can be used to iteratively identify the $s$-$t$ bridges in the residual graph after pushing \emph{unit} flow in the network.
	%After identifying the first bridge, 
	%The classical Ford Fulkerson algorithm~\cite{FordF56} can be used to identify the first $s$-$t$ bridge in the residual graph after pushing {\em unit} flow in the network.
	%Moreover, contracting the entire cut to $s$, one can continue finding the next $s$-$t$ bridge and so on.
	%Since $s$-$t$ bridges limit the maximum flow to {\em one}, the algorithm completes in linear time.
	%Further, it can be extended also works for minimum cuts of larger sizes with corresponding factor in the running time.

    \subsection{Problems Studied}
    \label{sec:contribution}
    
    Apart from the candidate set made up of all \emph{$s$-$t$ paths} (mentioned in the previous section), we will also consider two basic generalisations of it: the set of all \emph{$s$-$t$ trails} (i.e., walks from $s$ to $t$ which can repeat nodes, but not edges), and the set of all \emph{$s$-$t$ walks} (i.e., walks from $s$ to $t$, which can repeat both nodes and edges).
    We denote the problems of computing the maximal safe walks (in terms of alternating sequences of nodes and edges) for each of these problems as \textsc{MaxSafe} followed by \EPmodel, \ETmodel, and \EWmodel, respectively.
    
    In \Cref{f:s-t-safety}, we present examples for these problems.
    Neither of the coloured cycles can be used by an $s$-$t$ path, therefore the whole thick blue line is safe in \EPmodel{}.
    In \ETmodel{}, nodes can be reused and hence the red cycle (defined later as \emph{\forbiddentrail{}}), makes only the thick red lines as safe.
    In \EWmodel{}, \bridges{} can be reused as well and hence the green cycle (defined later as \emph{\forbiddenwalk{}}), makes only the thick green lines as safe.
    \begin{figure}[tbh]
	    \centering
	\includegraphics[width=.8\textwidth]{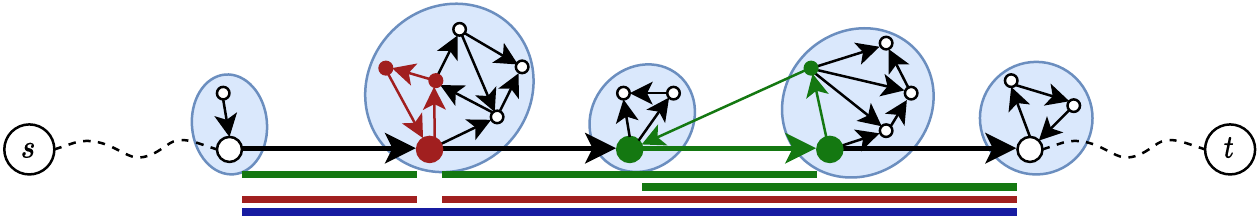}
	    \caption{Safe walks under different models for $s$-$t$-safety.
	    The figure shows a sequence of \bridges{} as bold arrows and their \islands{} as blue regions.
	    Thick blue, red and green lines show the answers to the \textsc{MaxSafe} \EPmodel{}, \ETmodel{} and \EWmodel{} problems, respectively.
	    \Forbiddentrails{} and \forbiddenwalks{} have been highlighted in red and green, respectively. %\todo{Hopefully draw.io fixes their pdf export soon}
	    }
	    \label{f:s-t-safety}
	\end{figure}
	
	An alternative way to look at these problems is to define safety in terms of only nodes, instead of walks.
	We define the \emph{node sequence} of a walk $W$ as the sequence (i.e., string) corresponding to $V$ obtained by reading the nodes of $W$ in order.
	A sequence of nodes is \emph{safe} if it is a substring of the node sequence of every walk in the candidate set.
	We denote the corresponding safety problems as \textsc{MaxSafe} followed by \VPmodel, \VTmodel, and \VWmodel, respectively.
	We alternatively refer to them as the \emph{$V$-visible} problems, whereas the natural case is referred as {\em $G$-visible}.
	The solutions to these $V$-visible problems can be obtained by simply leaving out the edges of maximal safe walks of $G$-visible problems.
	%But we get a more involved problem when we generalise to multigraphs (see \Cref{s:extensions}).
	But when extending to multigraphs, the node-centric trails problem becomes more involved, since it allows the repetition of a node sequence in a trail if and only if the node sequence is connected by multiedges (see \Cref{s:extensions}).
	
	Another dimension for generalising the above models is to assume also a \emph{visible} subset $X \subseteq V \cup E$ of nodes and edges and to define safety while looking only at the sequence of visible nodes and edges.
	We denote these problems as \textsc{MaxSafe} followed by \textsc{$X$-stPaths}, \textsc{$X$-stTrails}, and \textsc{$X$-stWalks}, respectively (or \emph{$X$-visible} problems).
	For example, we say that a sequence of nodes and edges in $X$ is \emph{safe} for the \textsc{$X$-stWalks} problem if it is a substring of the \emph{$X$-subsequence} (subsequence of elements in $X$) of each $s$-$t$ walk.
	
	\begin{table}[tbh]
		\caption{Computational complexity of the problems studied in this paper. $len(S)$ denotes the total length of the solution. (*) denotes the complexity for graphs without {\em multi-edges}, for {\em multigraphs} the complexity of \textsc{MaxSafe} \textsc{$V$-stPaths} is  $O(m+n+len(S))$.}
		\label{t:overview}
		\centering
		\begin{tabular}{|l|c c c|}
			\hline Visibility & \EP & \ET & \EW \\
			\hline 		
			$G$ & $O(m+n)$ & $O(m+n)$ & $O(m+n+len(S))$ \\
			$V$ & $O(m+n)$ & $O(m+n)$* & $O(m+n+len(S))$ \\
			$X$ & NP-hard & NP-hard & $O(m+n+len(S))$ \\ \hline
		\end{tabular}
	\end{table}
	In this paper we characterise the complexity of all nine \textsc{MaxSafe} problems for graphs, and later extend our results to multigraphs.
    %In this paper we characterise the complexity of all nine \textsc{MaxSafe} problems for graphs without multiedges, and show how to extend our results to multigraphs.
	See \Cref{t:overview} for a summary of these results.
	The $V$-visible problems are an interesting special case of the $X$-visible problems, as they are solvable in linear time even though they restrict visibility.
	This is a useful observation, as genome assembly problems are often modelled with a node-centric graph~\cite{TomescuMedvedev,Tomescu2017}.
	
	\subsection{Overview of our Approach}
	
	We solve all the linearly solvable \textsc{MaxSafe} problems with a similar algorithmic approach.
	Observe that a non-empty walk is uniquely defined by a sequence of edges.
	Therefore, we can simplify the $G$-visible and the $V$-visible problems to separately computing the maximal safe \emph{edge sequences} (analogue to \emph{node sequences}) and the maximal safe \emph{empty} walks (i.e. that consist of a single node).
	To obtain the solutions of the $G$-visible problems, it then suffices to complete the edge sequences with their corresponding nodes.
	And to obtain the solutions of the $V$-visible problems, observe that a sequence of nodes is safe if and only if it spells out the nodes of a safe walk.
	Therefore we take the solution of the corresponding $G$-visible problem and remove the edges.
	This separation has the advantage that for the more complex graph structures that govern the safety of non-empty walks we only need to consider edges, and adding back the nodes in the end is trivial.
	Moreover, if we are only interested in safe sequences of edges, we simply skip adding the nodes.
	Using a simple graph transformation, the \XW{} can also be solved by considering only edges.
	
	Observe that for a sequence of edges to be safe in our models, each edge needs to be safe on its own.
	Therefore, a safe sequence can only contain visible \bridges{}, which we compute as the first step.
	This \emph{bridge sequence} acts as the core of our solution, in the way that we can always describe the solution as a set of substrings of the bridge sequence, such that each \bridge{} is part of at least one maximal safe sequence.
	The bridge sequence (and similarly the \articulationsequence{}) can be computed with the classical min-cut algorithm~\cite{FordF56}.
	This algorithm was recently simplified for \bridges{} (or \articulationpoints{}) by Cairo et al.~\cite{CairoKRSTZ20}.
	
	%A simplified algorithm specialised on \bridges{} (or \articulationpoints{}) was recently published by Cairo et al.~\cite{CairoKRSTZ20}.
	
	The second step of our algorithms is to compute certain \emph{\forbiddenstructure{}s} (as shown in \Cref{f:s-t-safety}) that determine which substrings of the bridge sequence are maximal safe.
	For \EPmodel{} and \ETmodel{} and their $V$-visible counterparts, the \forbiddenstructure{}s do not cause solutions to overlap.
	Since the length of the bridge sequence is in $O(n)$ and the \forbiddenstructure{}s that define non-overlapping solutions are simple, we obtain the following result.
    
    \begin{theorem}
    \label{t:linearcases}
        Given a graph $G := (V, E)$ with $n$ nodes, $m$ edges and $s,t\in V$, there exist algorithms to compute \textsc{MaxSafe} $G$-visible (or $V$-visible) \EPmodel{} and \ETmodel{} in $O(m+n)$ time.
    \end{theorem}
    
    When extending to \EWmodel{}, we get more complex \forbiddenstructure{}s that both overlap themselves and cause the solutions to overlap.
    This poses two problems.
    First, there can be up to $O(n^2)$ of these \forbiddenstructure{}s (see \Cref{f:n2forbiddenwalks}), and the total length of the solution can be up to $O(n^2)$ (see \Cref{f:n2safe}).
    To handle the high amount of \forbiddenstructure{}s, we show that they can be reduced to a dominating set of size $O(n)$, which can be computed in  $O(m + n)$ time without computing all \forbiddenstructure{}s first.
    To handle the solution length, we make use of the bridge sequence $B$.
    We show that there are at most $O(|B|)$ maximal safe sequences, which allows us to represent the whole solution in a compact representation which can be unpacked in \emph{output-sensitive} linear time(time linear in output size).
    This representation consists of the bridge sequence and the start and end indices of each maximal safe sequence.
    Its total size is $O(|B|)$, and since each bridge is safe it never exceeds the total length of the solution.
    We show that this data structure can be computed in $O(m + n)$ time.
    
    \begin{theorem}
        \label{t:lineardatastructure}
        Given a graph $G := (V, E)$ with $n$ nodes, $m$ edges and $s,t\in V$, the corresponding bridge sequence $B$ and a subset $X \subseteq V \cup E$, there exist algorithms to compute a compact representation of the solution $S$ of \textsc{MaxSafe} $X$-visible \EWmodel{} of size $O(|B|)$ in $O(m+n)$ time, which can report the complete solution in $O(len(S))$ time.
    \end{theorem}
	
	In contrast to this, when considering {\em subset visibility} for \EPmodel{} and \ETmodel{}, the respective safety problems are NP-hard.
	We prove that by reducing from the \detour{} problem of finding a $u$-$v$ path passing through a third given node $w$  (see \Cref{p:detour}) which is known to be NP-hard.
	The reduction is possible, since the problems forbid edge repetitions and with subset visibility we can focus on the nodes of the \detour{} instance.
    
    \begin{theorem}[restate = nphardness, name = ]
	    \label{t:np-hardness}
    	The \textsc{MaxSafe} $X$-visible \EPmodel{} and \ETmodel{} problems are NP-hard, even when deciding the safety of a sequence of just two elements of $X$ and restricting $X$ to contain only nodes or only edges.
	\end{theorem}

    \subparagraph{Organisation of the paper.} We describe the results in an incremental manner, gradually building upon the previous solution to solve harder problems. In the remaining subsections of this section we define our notation and describe some preliminary results including the \bridge{} algorithm.
    In \Cref{s:safetyEP} we describe how the \bridge{} algorithm can be expanded to solve \EP{} and \ET{}.
    In \Cref{s:min-forb-walks} we describe another algorithm that can be used on top of the \bridge{} algorithm to solve \EW{}.
    As stated above, solving these problems also solves all $V$-visible problems.
    In \Cref{s:subset-visibility} we describe how that algorithm can be expanded to solve \XW{}.
    In \Cref{s:extensions} we describe how to extend our results to multigraphs.
    In \Cref{s:conclusions} we review our results.

    %\todo{Sebastian, Shahbaz: better sell our algorithms in the above. Is there something especially interesting, for example, we incrementally build the algorithms from simple to hard---instead of just giving directly the most general and hardest one, is there some surprising fact about some characterisation, a non trivial way in which we use the characterisation (e.g., we can start already with just $O(n)$ cutting loops), etc}
	
	\subsection{Notation}
	
	As defined above, we assume a fixed graph $G := (V, E)$, where $V$ is a set of $n$ nodes and $E$ a set of $m$ edges.
	Furthermore, we assume two nodes $s, t \in V$ are given.
	A \emph{graph} is directed and may include loops, but not multiedges.
	A graph with multiedges is a \emph{multigraph}.
	Given a set of nodes and edges $X \subseteq V \cup E$, or a single node or edge $X$, then $G[X]$ denotes its induced subgraph and $G - X$ is the result of removing from $X$ all edges and nodes (together with their incident edges).
	If $X$ contains only edges, we may also write $G \setminus X$.
	
	%Let $G\setminus X$ and $G - X$ denote the result of removing all edges from $X$, and all nodes from $X$ together with their incident edges, respectively.
	%$G-X\setminus X$

	Given an edge $e=(u,v)$, $\head{}(e)=v$ denotes its \emph{head} and $\tail{}(e)=u$ denotes its \emph{tail}.
	Given a sequence, a \emph{subsequence} is obtained by removing arbitrary elements, while a \emph{substring} is obtained by removing a prefix and a suffix (both possibly empty).
	A sequence $W := (v_1, e_1, v_2, \dots, v_{|W|}, e_{|W|}, v_{|W|+1})$ of nodes $v_i$ and edges $e_i$ is a \emph{$v_1$-$v_{|W|+1}$ walk} (or simply \emph{walk}) if $v_i = \tail{}(e_i)$ and $v_{i+1} = \head{}(e_i)$ for all $i \in \{1, \dots, |W|\}$.
	Its subsequence of only elements from a set $X$ is called its \emph{$X$-subsequence}; if $X = V$ it is called its \emph{node sequence}, and if $X = E$ it is called its \emph{edge sequence}. %\sebastian{We defined this already above, should we redefine it here?}.
	A walk $W$ is a \emph{$v_1$-$v_{|W|+1}$ trail} (or simply \emph{trail}) if it repeats no edge, and it is a \emph{$v_1$-$v_{|W|+1}$ path} (or simply \emph{path}) if it additionally repeats no node, except that $v_1$ may equal $v_{|W|+1}$, in which case it is a \emph{cycle}.
	A walk $W$ is \emph{empty} if it contains only one node and \emph{non-empty} otherwise.
	A walk $W$ \emph{contains} a sequence of edges (or nodes), if that sequence is a substring of its edge sequence (or node sequence).
	
	The \emph{node expansion} of a node $v$ is an operation that transforms $G$ into a graph $G'$ by adding a node $v'$ and an edge $e_v$ from $v$ to $v'$ and moving all out-edges from $v$ to $v'$.
	We call $e_v$ the \emph{internal edge} of $v$.
	
	\subsection{Preliminaries}
	\newcommand{\BrgS}{B}
    \newcommand{\brgS}{b}
	\newcommand{\ComB}{{\cal C}}
	\newcommand{\comB}{C}
	
	 %% Definition
	Let $\BrgS=\{\brgS_1,\brgS_2,...,\brgS_{|\BrgS|}\}$ be the set of $s$-$t$ bridges of $G$. By definition, for all $\brgS_i\in \BrgS$ there exists no path from $s$ to $t$ in $G\setminus \brgS_i$ (see \Cref{fig:bridgeCSa}), and all \bridges{} in $\BrgS$ appear on every $s$-$t$ path in $G$.
	Further, the \bridges{} in $\BrgS$ demonstrate the following interesting property.% are visited in the same order by every $s$-$t$ path in $G$.
	
	\begin{lemma}       
	The \bridges{} in $\BrgS$ are visited in the same order by every $s$-$t$ path in $G$.
	\label{lem:bridgeOrder}
	\end{lemma}
    \begin{proof}
    It is sufficient to prove that for any $\brgS_i\in  \BrgS$, all $b_j\in \BrgS$ (where $j\neq i$), can be categorised into those which are always visited before $\brgS_i$ and those that are always visited after $\brgS_i$ irrespective of the $s$-$t$ path chosen in $G$. Consider the graph $G\setminus \brgS_i$, observe that every such $\brgS_j$ is either reachable from $s$, or can reach $t$. It cannot fall in both categories as it would result in an $s$-$t$ path in $G\setminus \brgS_i$, which violates $\brgS_i$ being an \bridge{}. Further, it has to be in at least one category by considering any $s$-$t$ path of $G$, where $\brgS_i$ appears either between $s$ and $\brgS_j$ or between $\brgS_j$ and $t$. Hence, those reachable from $s$ in $G\setminus \brgS_i$ are always visited before $\brgS_i$, and those able to reach $t$ in $G\setminus \brgS_i$ are always visited after $\brgS_i$, irrespective of the $s$-$t$ path chosen in $G$.  
	\end{proof}

\begin{figure}[tbh!]
\centering
\begin{subfigure}[t]{0.5\textwidth}
\centering
\includegraphics[width=\textwidth]{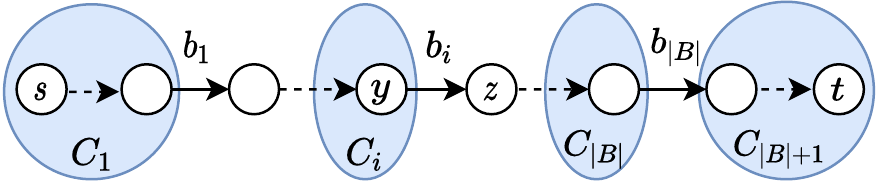}
\caption{Bridge sequence $\BrgS = \{\brgS_1,\dots,\brgS_{|B|}\}$ and corresponding bridge components $\ComB=\{\comB_1,\dots,\comB_{|B|+1}\}$.}
\label{fig:bridgeCSa}
\end{subfigure}\hfill
\begin{subfigure}[t]{0.45\textwidth}
\centering
\includegraphics[width=.8\textwidth]{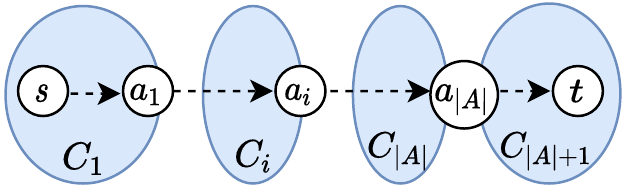}
\caption {Articulation sequence $A =\{a_1,\dots,a_{|A|}\}$ and its components $\ComB=\{\comB_1,\dots,\comB_{|A|+1}\}$ .}
\label{fig:bridgeCSb}
\end{subfigure}\hfill
\caption{Bridge and articulation sequences with their components along $s$-$t$ path. Recall that $\comB_i$ is the new part of $G$ reachable from $s$ in $G\setminus \brgS_i$ (or $G- a_i$) in comparison to $G\setminus \brgS_{i-1}$ (or $G- a_{i-1}$).}
\end{figure}

	%% Definition Islands
	Thus, abusing the notation, we define $\BrgS$ to be a \emph{sequence of \bridges{}} ordered by their visit time on any $s$-$t$ path.
	Such a bridge sequence $\BrgS$ implies an increasing part of the graph being reachable from $s$ in $G\setminus\brgS_i$, as $i$ increases.
	We thus divide the graph reachable from $s$ into \emph{bridge components} $\ComB=\{\comB_1,\comB_1,...,\comB_{|\BrgS|+1}\}$, where $\comB_i$ (for $i\leq |\BrgS|$) denotes the part of the graph that is reachable from $s$ in $G\setminus \brgS_i$ but was not reachable in $G\setminus \brgS_{i-1}$ (if any).
	Additionally, for notational convenience we assume $\comB_{|B|+1}$ to be the part of the graph reachable from $s$ in $G$, but not in $G \setminus \brgS_{|B|}$ (see \Cref{fig:bridgeCSa}).
	Since bridge components are separated by \bridges{}, every $s$-$t$ path enters $\comB_i$ at a unique node ($\head{}(\brgS_{i-1})$ or $s$ for $\comB_1$) referred as its \textit{entrance}.
	Similarly, it leaves $\comB_i$ at a unique node ($\tail{}(\brgS_i)$ or $t$ for $\comB_{|B|+1}$) referred as its \textit{exit}.

	Similarly, the $s$-$t$ \emph{articulation points} are defined as the set of nodes $A\subseteq V$, such that removal of any \articulationpoint{} in $A$ disconnects all $s$-$t$ paths in $G$. Thus, $A=\{a_1,a_2,...,a_{|A|}\}$ is a set of nodes such that $\forall a_i\in A$ there exist no path from $s$ to $t$ in $G- a_i$. 
	The \articulationpoints{} in $A$ also follow a fixed order in every $s$-$t$ path (like \bridges{}), so $A$ can be treated as a sequence and it defines the corresponding components~$\ComB$ (see \Cref{fig:bridgeCSb}). Note that the {\em entrance} and {\em exit} of an articulation component $\comB_i$ are the preceding and succeeding \articulationpoints{} (if any), else $s$ and $t$ respectively.
	
	The $s$-$t$ bridges and articulation points along with their component associations can be computed in linear time, using either flows as described above, or the referenced simplification.
	
    %The classical Ford Fulkerson's algorithm~\cite{FordF56} to compute a Minimum Cut in a graph can be simplified when the minimum cut is known to be of {\em unit} size. Further, commonly known extensions can be used to find all such minimum cuts of {\em unit} size. For the sake of completeness we present this algorithm in a simplified manner independent of the theory of maximum flows in \Cref{ss:bridgealgo}. The algorithm can be essentially described as a {\em forward search} from $s$ to $t$ avoiding an arbitrary $s$-$t$ path, which is {\em interrupted} at the $s$-$t$ bridges (or \articulationpoints{}).
	\begin{theorem}[restate = linearst, name = \cite{FordF56,CairoKRSTZ20}]
	\label{t:linearst}
	Given a graph $G := (V, E)$ with $n$ nodes, $m$ edges and $s,t\in V$, there exists an algorithm to compute all $s$-$t$ bridges and $s$-$t$ articulation points, along with their component associations, in $O(m+n)$ time.
	\end{theorem}
	
	\section{Safety for \textnormal{\EPmodel{}} and \textnormal{\ETmodel{}}}
	\label{s:safetyEP}
	
	The \bridge{} algorithm (\Cref{t:linearst}) is the main building block for proving \Cref{t:linearcases}.
	Recall that we simplified the corresponding problems to only finding the maximal safe edge sequences.
	The solution to \EP{} directly follows from the \bridge{} algorithm.
	
	Observe that for two \bridges{} to form a safe sequence, they need to be adjacent.
	In the \EPmodel{} model, this is also sufficient, because visiting any other edge from the intermediate node would repeat the node to reach the latter edge (see the thick blue line in \Cref{f:s-t-safety}).
	Therefore, we get the following characterisation.
	
	\begin{theorem}[restate = charep, name = Safety for \EPmodel{}]
	    \label{t:charesp}
		A substring of the \bridgesequence{} is safe under the \EPmodel{} model, if and only if each consecutive pair of edges is adjacent.
	\end{theorem}
	\begin{proof}
		$(\Rightarrow)$ Let $L$ be a substring of the \bridgesequence{} that is safe under the \EPmodel{} model.
		Then $L$ is a subpath of an arbitrary $s$-$t$ path.

		$(\Leftarrow)$ Let $L := (e_1, \dots, e_{|L|})$ be a substring of the \bridgesequence{} such that each consecutive pair of edges is adjacent.
		Let $W$ be an $s$-$t$ path and $W_E$ its $E$-subsequence.
		We prove that $L$ is a substring of $W_E$ by induction.
		For the base case, note that $W_E$ contains $e_1$ since $L$ is made of strong $s$-$t$ bridges.
		For the inductive step, assume that $W_E$ contains $(e_1, \dots, e_i)$ as substring.
		Since $L$ is a substring of the \bridgesequence{}, $W$ contains $e_{i+1}$ after $e_i$.
		And since $\head{}(e_i) = \tail{}(e_{i+1})$ and $W$ is a path, it contains $e_{i+1}$ immediately after $e_i$.
	\end{proof}
	
	Since each $s$-$t$ path is an $s$-$t$ trail, adjacency is still necessary for safety of \ETmodel{}, but not sufficient.
	In \Cref{f:s-t-safety}, the $s$-$t$ trail that uses the red cycle breaks the safe \EPmodel{} (thick blue line).
    Thus, such a red cycle or the non-adjacency of \bridges{} makes a {\em \forbiddentrail{}}.
	
	\begin{definition}[\ForbiddenTrail{}]
		\label{d:cuttingloop}
		A \emph{\forbiddentrail{}} is a path that connects two consecutive \bridges{} $b_i, b_{i+1}$ for $i \in \{1, \dots, |B| - 1\}$ without using either of them.
	\end{definition}

\begin{figure}[tbh!]
\centering
\begin{subfigure}[b]{0.47\textwidth}
\centering
\includegraphics[scale=1]{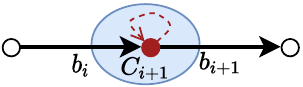}
\caption{A \forbiddentrail{} between adjacent \bridges{}.}
\label{f:forbiddentrails:adjacent}
\end{subfigure}~
\begin{subfigure}[b]{0.52\textwidth}
\centering
\includegraphics[scale=1]{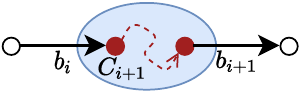}
\caption {A \forbiddentrail{} between non-adjacent \bridges{}.}
\label{f:forbiddentrails:separate}
\end{subfigure}
\caption{Examples for \forbiddentrails{} (red) 
	    %\hint{ending at different nodes ( non-adjacency of consecutive bridges) or same nodes (cycle avoiding the bridges)}. 
	    The bold edges are \bridges{} and the blue regions mark \islands{}. %\hint{may not be required}
	    }
\label{f:forbiddentrails}
\end{figure}

% 	\begin{figure}[tbh]
% 	    \centering
	   % \subcaptionbox{A \forbiddentrail{} between two non-adjacent \bridges{}. \label{f:forbiddentrails:separate}}
	   % {	  
	   % \includegraphics[]{s-t-safety-forbidden-trails-2c}
	   % %\includegraphics[]{s-t-safety-forbidden-trails-2}
	   % }
	   % % \hfill
	   % \hspace{2cm}
	   % \subcaptionbox{A \forbiddentrail{} between two adjacent \bridges{}.  \label{f:forbiddentrails:adjacent}}
	   % {
	   % \includegraphics[]{s-t-safety-forbidden-trails-1}
	    %\includegraphics[]{s-t-safety-forbidden-trails-1}
	   % }
% 	    \caption{Examples for \forbiddentrails{} (red) 
% 	    %\hint{ending at different nodes ( non-adjacency of consecutive bridges) or same nodes (cycle avoiding the bridges)}. 
% 	    The bold edges are \bridges{} and the blue ellipses mark \islands{}. %\hint{may not be required}
% 	    \todo{Hopefully draw.io fixes their pdf export soon}\hint{not cropped correctly}
% 	    }
% 	    \label{f:forbiddentrails}
% 	\end{figure}
	
	Any path $P$ which is a \forbiddentrail{} for a sequence $(b_i, b_{i+1})$ lies completely within the \island{} $C_{i+1}$.
	Otherwise, $P$ would also contain either $b_i$ or $b_{i+1}$, which is not allowed.
	Note also that $P$ is a cycle if $\head{}(b_i) = \tail{}(b_{i+1})$.
	See \Cref{f:forbiddentrails:adjacent} for an illustration of a circular \forbiddentrail{} and \Cref{f:forbiddentrails:separate} for a non-circular one.
	We get the following characterisation.

% 	The existence of a \forbiddentrail{} $P$ between $b_i$ to $b_{i+1}$ allows any $s$-$t$ trail to reach $b_i$ without using $P$, then use $P$ to reach $b_{i+1}$ and continue to $t$, without repeating an edge. Thus, an $s$-$t$ trail can choose multiple paths to reach $\brgS_{i+1}$ from $\brgS_i$ which disproves the safety of any trail containing both $b_i$ and $b_{i+1}$. On the other hand, whenever a consecutive pair of \bridges{} $b_i$ and $b_{i+1}$ is not safe in an $s$-$t$ trail, then the part of any $s$-$t$ trail between those two \bridges{} proves to be the corresponding  \forbiddentrail{} which made it unsafe. \hint{is the para until here a summary of the proof? if so to avoid repetition comment everything till here, and merge the following line with previous para.}
% 	Thus, we get the following characterisation.
	
	\begin{theorem}[restate = charet, name = Safety for \ETmodel{}]
		\label{t:charesw}
		A substring of the \bridgesequence{} is safe under the \ETmodel{} model, if and only if it has no \forbiddentrail{}.
	\end{theorem}
	\begin{proof}
% 		We proceed as in the proof of \Cref{t:charesp}.
	
		$(\Rightarrow)$ Let $L$ be a substring of the \bridgesequence{} that is safe under the \ETmodel{} model.
		Assume for a contradiction that $L$ contains a \forbiddentrail{} $P$ from $\head{}(e_i)$ to $\tail{}(e_{i+1})$ for some $i \in \{1, \dots, |L| - 1\}$.
		As such, $P$ is completely inside a \island{} $C$ (the one with exit $\tail{}(e_{i+1})$).
		It holds that any $s$-$t$ trail $W$ (which contains $L$ as substring, because $L$ is safe) does not contain any edge from $C$.
		As a result, we can insert $P$ in $W$ between $\head{}(e_i)$ and $\tail{}(e_{i+1})$ to obtain an $s$-$t$ trail $W'$ that does not contain $L$.
		This contradicts the safety of $L$.
		
		$(\Leftarrow)$ Let $L := (e_1, \dots, e_{|L|})$ be a substring of the \bridgesequence{} that has no \forbiddentrail{}.
		Let $W$ be an $s$-$t$ trail and let $W_E$ be the $E$-subsequence of $W$.
		We prove that $L$ is a substring of $W_E$ by induction.
		Note that $W_E$ contains $e_1$, by assumption.
		For the inductive step, suppose that $W_E$ contains $(e_1, \dots, e_i)$ as substring.
		Since $L$ is a substring of the \bridgesequence{}, $W_E$ contains $e_{i+1}$ after $e_i$.
		And since $L$ has no \forbiddentrail{}, each non-empty $\head{}(e_i)$-$\tail{}(e_{i+1})$ path $P$ contains $e_i$ or $e_{i+1}$.
		But $e_i$ is already used by $W$ on the way to $\head{}(e_i)$, and $e_{i+1}$ needs to be used to reach $t$ from $\tail{}(e_{i+1})$.
		So no such $P$ can be a subwalk of $W$, and thus $e_i$ is immediately followed by $e_{i+1}$ in $W_E$.
	\end{proof}
	
	Using these characterisations, the \bridge{} algorithm can be directly used to solve the problems. After computing the \bridges{} and \islands{},  the \bridgesequence{} is split between non-adjacent pairs of consecutive \bridges{}, solving \EP{}.
	For \ET{}, the residual pairs of consecutive \bridges{} are checked for a \forbiddentrail{} by checking if there is an incoming edge to $\tail{}(b_{i})$ from a node in $C_{i}$.
	%The \forbiddentrail{} ending at an \bridge{} $b_i$ exists if there is an incoming edge to $\tail{}(b_{i})$ from a node in $C_{i}$.
	Thus, both problems can be solved in $O(m + n)$ time, proving \Cref{t:linearcases}.
	
	\section{Safety for \textnormal{\EWmodel{}}}
	\label{s:min-forb-walks}
	
	Unlike the previous problems, solving \EW{} requires another algorithmic building block.
%	Recall that it suffices to focus on safe edge sequences.
	Again, since each $s$-$t$ trail is an $s$-$t$ walk, the absence of \forbiddentrails{} is necessary for safety for \EWmodel{}, but not sufficient. In \Cref{f:s-t-safety}, an $s$-$t$ walk using the green cycle breaks the thick red line that is safe in \ETmodel{}.
	Thus, such a green cycle or a \forbiddentrail{} makes a {\em \forbiddenwalk{}}.
	
	%\subparagraph{Characterisation for \textnormal{\EWmodel{}}.} 
	%Each $s$-$t$ trail is an $s$-$t$ walk, and therefore the conditions of \Cref{t:charesw} are necessary in the \EWmodel{} model as well.
	%Since a walk can repeat edges the \forbiddenstructure{} in this model can contain \bridges{}.
	%See the green cycle in \Cref{f:s-t-safety} for an example.
	%We formally define these as \forbiddenwalks{} below.
	
	\begin{definition}[\ForbiddenWalk{}]%, $F$-Cutting Walk]
		\label{d:forbiddenwalk}
		A \emph{\forbiddenwalk{}} is a path that connects two consecutive \bridges{} $e_i, e_{i+1}$ from a substring $L$ of the bridge sequence for $i \in \{1, \dots, |L| - 1\}$ without using the first or last edge from $L$.
	\end{definition}

	\begin{figure}[tbh]
	    \centering
	    \includegraphics[scale=0.8]{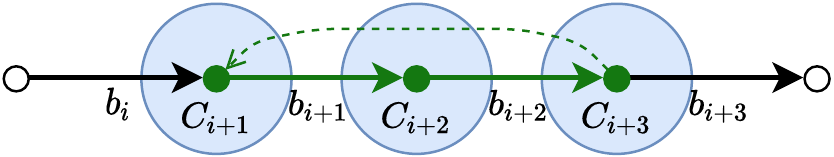}
	    \caption{An example for a \forbiddenwalk{}, highlighted in green.
	    The bold edges are \bridges{} and the blue regions mark \islands{}.}
	    \label{f:forbiddenwalks}
	\end{figure}

    \Forbiddenwalks{} can stay within a single \island{} (like \forbiddentrails{}, recall \Cref{f:forbiddentrails}), but can also include multiple \islands{} (see \Cref{f:forbiddenwalks}).

    \subparagraph{Characterisation for \textnormal{\EWmodel{}}.} Since a \forbiddenwalk{} $P$ contains neither the first or last edge of its corresponding substring $L$ of the bridge sequence, it contains no prefix or suffix of $L$.
    Therefore, if $P$ is inserted into an $s$-$t$ path that contains $L$, then the result contains a prefix and a suffix of $L$ that together spell $L$, but that are interrupted by a subwalk that neither completes a prefix or suffix nor contains $L$ itself.
    Thus, $P$ contradicts the safety of $L$.
    On the other hand, if it is possible to construct an $s$-$t$ walk $W$ that does not contain $L$, then $W$ contains a last occurrence of $e_1$, the first \bridge{} of $L$.
    After this last occurrence of $e_1$, $W$ contains a non-empty subwalk $W'$ between a pair of consecutive \bridges{} $e_i, e_{i+1}$ from $L$ that does not contain $e_{i+1}$ or $e_1$.
    Therefore, $W'$ is a trail breaker for $L$.
    Resulting, we get a characterisation, using the following two lemmas required for its proof.
    
    \begin{lemma}
    	A non-zero bridge length minimal \forbiddenwalk{} with the bridge sequence starting at $\brgS_i$ and ending at $\brgS_{j}$ ($i\leq j$), satisfies the following conditions
    	\begin{enumerate}[(a)]
    	    \item There exists a {\em backward edge} $e$ from $C_{j+1}$ to $C_{i}$.
    	    \item The exit of $C_i$ is reachable from $\head{}(e)$ (i.e. marked).
    	\end{enumerate}
    	\label{l:backwardsedgesareforbiddenwalks}
	\end{lemma}
	    
	\begin{proof}
        We first prove that the minimal \forbiddenwalk{} contains a single backward edge by contradiction. Assume it contains multiple backward edges $e_1,...,e_k$ in order. Now consider a \forbiddenwalk{} using only the backward edge $e_1$, clearly its bridge sequence is a substring of the bridge sequence of the original \forbiddenwalk{} and hence dominates it.
        
        Now, to complete the \forbiddenwalk{} we require $\tail{}(e)$ to be reachable from $\head{}(b_{j})$ and $\tail{}(b_{i})$ to be reachable from $\head{}(e)$. Since $e$ originates from $C_{j+1}$ with its entrance $\head{}(b_j)$ we can reach $\tail{}(e)$. The later case is ensured by the algorithm by marking only those nodes in $C_i$ which can reach $\tail{}(b_i)$. 
        Finally, if the edge $e$ ends in $C_i$, it cannot start in a different \island{} $\comB_k$. If $k < i$ then \forbiddenwalk{} would contain an \bridge{} before $b_i$; symmetrically also $k \not\geq j$ holds.  If $k \in \{i + 1, \dots, j\}$, then there exists a \forbiddenwalk{} that starts in $b_i$ and ends in $b_{k-1}$, which contradicts the minimality of the \forbiddenwalk{}, ensuring that $e$ starts in $C_j$.
	\end{proof}
	
    \begin{lemma}
        Given the ordered set of minimal \forbiddenwalks{} $\mathcal{P} = \{P_1,P_2,...,P_k\}$ in the \EWmodel{}, a maximal safe sequence begins with $b_1$ and ends with $b_{|B|}$ if $\mathcal{P} = \emptyset$, and otherwise either:
        \begin{enumerate}[(a)]
            \item starts with $b_1$ and ends with the end of $P_1$, or \label{l:maxforbiddenwalkstructure:first}
            \item ends with $b_{|B|}$ and starts with the start of $P_k$, or \label{l:maxforbiddenwalkstructure:last}
            \item starts with the start of $P_i$ and ends with the end of $P_{i+1}$ for some $i \in \{1, \dots, k - 1\}$. \label{l:maxforbiddenwalkstructure:middle}
        \end{enumerate}
        \label{l:maxforbiddenwalkstructure}
    \end{lemma}
    \begin{proof}
        The case where $\mathcal{P} = \emptyset$ is trivial.
        Otherwise, let $L$ be a maximal safe sequence.
        If $L$ starts with an \bridge{} other than $b_1$ that is not the start of a minimal \forbiddenwalk{}, then it can be extended to the left without becoming unsafe, contradicting its maximality.
        By symmetry, $L$ cannot end with an \bridge{} other than $b_{|B|}$ that is not the end of a minimal \forbiddenwalk{}.
        Therefore, the start and end \bridges{} considered in (\labelcref{l:maxforbiddenwalkstructure:first}) --  (\labelcref{l:maxforbiddenwalkstructure:middle}) are sufficient.
        It remains to prove that the pairings are correct.
        
        \begin{enumerate}[(a)]
            \item If $L$ starts with $b_1$, then it needs to end no later than the end of $P_1$, since otherwise it would be proven unsafe by $P_1$.
            It can end no earlier, since no \forbiddenwalk{} ends before $P_1$.
        
            \item If $L$ ends with $b_{|B|}$, then by symmetry with (\labelcref{l:maxforbiddenwalkstructure:first}) it starts with $P_k$.
            
            \item
            If $L$ starts with the start of $P_i$ for some $i \in \{1, \dots, k - 1\}$, then it cannot end after the end of $P_{i+1}$, since otherwise it would be proven unsafe by $P_{i+1}$.
            Furthermore, if it ends in the end of $P_{i+1}$, then it is neither proven unsafe by $P_i$ nor by $P_{i+1}$.
            And, since $\mathcal{P}$ is ordered, it is also not proven unsafe by another \forbiddenwalk{}.
            As such, $L$ ends with the end of $P_{i+1}$. \qedhere
        \end{enumerate}
    \end{proof}
	
	\begin{theorem}[restate = charew, name = Safety for \EWmodel{}]
		\label{t:charew}
		A substring of the \bridgesequence{} is safe under the \EWmodel{} model, if and only if it has no \forbiddenwalk{}.
	\end{theorem}
	\begin{proof}
% 		We proceed as in the proofs of \Cref{t:charesp,t:charesw}.
	
		$(\Rightarrow)$ Let $L$ be a substring of the \bridgesequence{} that is safe under the \EWmodel{} model.
		If there is a \forbiddenwalk{} for $L$, then from an $s$-$t$ walk $W$ an $s$-$t$ walk $W'$ can be constructed by inserting the \forbiddenwalk{} into every occurrence of $L$.
		But then, $L$ is not a substring of $W'$.
		
		$(\Leftarrow)$ Let $L := (e_1, \dots, e_{|L|})$ be a substring of the \bridgesequence{} that has no \forbiddenwalk{}.
		Let $W$ be an $s$-$t$ walk and let $W_E$ be its $E$-subsequence.
		Since $L$ is a substring of the \bridgesequence{}, $W_E$ contains a substring $W'_E$ that starts from the last occurrence of $e_1$ and ends in the first occurrence of $e_{|L|}$ after that, and $L$ is a subsequence of $W'_E$.
		We prove that $L$ is a prefix of $W'_E$ by induction.
		By definition, $W'_E$ starts with $e_1$.
		For the inductive step, suppose that $(e_1, \dots, e_i)$ is a prefix of $W'_E$.
		By definition of $W'$, none of its $\head{}(e_i)$-$\tail{}(e_{i+1})$-subwalks contain $e_1$ or $e_{|L|}$.
		Therefore, since $L$ does not have any \forbiddenwalk{}, $e_{i+1}$ immediately follows $e_i$ in $W'_E$.
	\end{proof}
    
    \section{Computing \ForbiddenWalks{} Efficiently}
    \label{s:efficientforbiddenwalks}
    
	Since \forbiddenwalks{} can span over multiple \islands{}, their structure is more complex than that of \forbiddentrails{}.
	But as with \forbiddentrails{}, 
	the only edges of the \forbiddenwalk{} that are relevant, are the \bridges{}.
	%it is not relevant which particular edges a \forbiddenwalk{} uses, except for the \bridges{}.
	Moreover, since \forbiddenwalks{} cannot skip \bridges{} they always correspond to a {\em substring} of the \bridgesequence{}, where \forbiddenwalks{} with same  {\em substrings} are equivalent.
	We refer to this substring as the \emph{bridge sequence} of the \forbiddenwalk{}, and call its first edge the \emph{start} and its last edge the \emph{end} of the \forbiddenwalk{}.
	If a \forbiddenwalk{} that connects $b_i$ to $b_{i+1}, i \in \{1, \dots, |B| - 1\}$ contains no \bridge{} (i.e. it is equivalent to a \forbiddentrail{}), then we call $b_{i+1}$ its start and $b_i$ its end.
	We call the amount of \bridges{} in the bridge sequence of a \forbiddenwalk{} its \emph{bridge length}.
    
	\subparagraph{Minimal \forbiddenwalks{}.} 
	
	\begin{figure}
	    \centering
	    \includegraphics[scale=0.8]{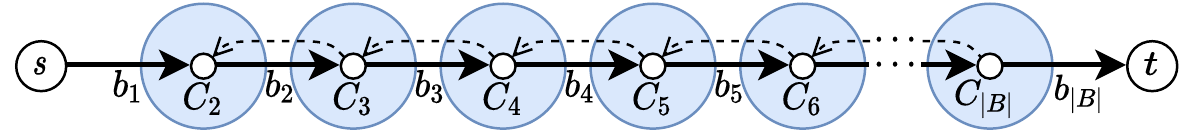}
	    \caption{A graph with $\Theta(n^2)$ \forbiddenwalks{}.
	    The bold edges are \bridges{} and the blue regions mark bridge components.
	    \label{f:n2forbiddenwalks}}
	\end{figure}
	
    Given $s$ and $t$, in the worst case, a graph may contain up to $\Theta(n^2)$ different \forbiddenwalks{}; see \Cref{f:n2forbiddenwalks} for an example.
	In this graph, there are $|B| - i - 1$ \forbiddenwalks{} of bridge length $i$, for each $i \in \{0, \dots, |B| - 2\}$.
	However, if the bridge sequence of a \forbiddenwalk{} is a substring of the bridge sequence of another \forbiddenwalk{}, every walk proven unsafe by the {\em later} is also proven unsafe by the {\em former}, i.e., the
	{\em former} dominates the {\em later}.
	Hence, computing all \forbiddenwalks{} is wasteful and we only focus on inclusion-minimal \forbiddenwalks{}, referred as \emph{minimal \forbiddenwalks{}} which are only dominated by themselves.
	
	With this notion of domination, it suffices to compute the set of minimal \forbiddenwalks{} in the graph to exclude all the unsafe walks.
    Note that at most one minimal \forbiddenwalk{} starts and ends at each \bridge{}, otherwise one would dominate the other.
    %Note that in each node at most one minimal \forbiddenwalk{} starts \hint{and ends}, since otherwise one would dominate the other.
	Therefore, there are at most $O(|\BrgS|)$ different minimal \forbiddenwalks{} for a \bridgesequence{} $\BrgS$.
	We now describe how to compute these in linear time.
	
	\subparagraph{Algorithm.}
	
	\begin{algorithm}[tbh]
		\caption{\textsc{Minimal \ForbiddenWalks{}}}
		\label{a:forbiddenwalks}
		\KwIn{Graph $G := (V, E)$, with $s$-$t$ \bridgesequence{} $B$ and its components in $\ComB$.
		%:= (b_1, \dots, b_{|B|})$, \islands{} $C_1, \dots, C_{|B|+1}$
		}
		\KwOut{Set of non-zero length \forbiddenwalks{} that are minimal from their starts.
		%with respect to their starts, ignoring bridge length zero \forbiddenwalks{}
		}
		\DontPrintSemicolon
		\vspace*{0.2em}
		\BlankLine

        \tcc*[l]{Stage one}
		%\tcc*[l]{Initialise ends of \forbiddenwalks{}, carrying over \forbiddentrails{}}
		\ForAll(\tcp*[f]{initialise  \forbiddenwalks{} }){$i \in \{2, \dots, |B|\}$}{
		    \lIf{there exists a \forbiddentrail{} from $b_{i-1}$ to $b_i$}{
		        $end[i] \gets i - 1$
		    }\lElse{
		        $end[i] \gets |B|$\tcp*[f]{signifies empty starting from $i$}
		    }
		}
        \BlankLine

		\lForAll{$i \in \{2, \dots, |B| - 1\}$}{
		    Mark nodes in $\comB_i$ reverse reachable from $\tail{}(b_i)$
		    %Mark all nodes within $\comB_i$ that are reverse reachable from $\tail{}(b_i)$\;
		}
        \BlankLine
		\ForAll{$j \in \{2, \dots, |B| - 1\}$ \label{a:forbiddenwalks:cycleloop}}{
			%\tcc*[l]{Find circular \forbiddenwalks{} that end in $b_j$}
			\ForAll(\tcp*[f]{circular \forbiddenwalks{} ending with $b_j$}){$u \in \comB_{j+1}$ \label{a:forbiddenwalks:searchforenode}}{
				\ForAll{$(u, v) \in E: v \in \comB_i, 2 \leq i \leq j$ \label{a:forbiddenwalks:searchfore}}{
					\If{$v$ is marked}{
						$end[i] \gets \min\{end[i], j\}$\;
					}
				}
			}
		}
        \BlankLine

		\tcc*[l]{Stage two}
		$min \gets |B|$\tcp*[r]{minimum end seen so far}
		\ForAll{$i \in \{2, \dots, |B|\}$ in reverse}{
		  %  \If{$end[i] = |B|$}{
		  %      \Continue\tcp*[r]{skip positions in which no \forbiddenwalk{} starts}
		  %  }
		  %  \If{$end[i] \geq min$}{
		  %      $end[i] \gets |B|$\tcp*[r]{remove \forbiddenwalks{} that are dominated by others}
		  %  }\Else{
		   \If(\tcp*[f]{ignore dominated \forbiddenwalks{}}){$end[i] < min$}{
		        Add $(i,end[i])$ to \textsc{Sol}\;
		        $min \gets \min\{min, end[i]\}$\tcp*[r]{\forbiddenwalk{} is new leftmost end}
		    }
		Return \textsc{Sol}\;
		}
	\end{algorithm}
	
	Using \bridges{} and \forbiddentrails{} (zero bridge length \forbiddenwalks{}) computed earlier, we now compute the minimal \forbiddenwalks{} of non-zero bridge length in two stages.
	First, we compute the $O(|\BrgS|)$ \forbiddenwalks{} that are minimal with respect to their starts. Then we remove the dominated \forbiddenwalks{} to get the globally minimal \forbiddenwalks{}.
	
	In the first stage, we start by performing backwards traversals from $t$ and the tail of each \bridge{} that stay within the \island{} they started.
	This way, we {\em mark} each node that is reverse reachable from the exit of its \island{}.
	Now, all the \forbiddenwalks{} minimal from their start correspond to {\em backward edges} $e$ of the following form: an edge from $C_j$ to $C_i$ ($i<j$) where $\head{}(e)$ is marked. Intuitively, a minimal \forbiddenwalk{} contains a single backward edge, because if it has multiple backward edges its minimality would be disproven by one of its backwards edges. The marked nodes ensure that \forbiddenwalk{}  is completed by reaching the start (see \Cref{l:backwardsedgesareforbiddenwalks}). Hence, we iterate over all such edges and maintain the dominating \forbiddenwalk{} starting from each $s$-$t$ bridge.
	
	In the second stage, we traverse the \forbiddenwalks{} in reverse order of their starts and remove those that do not end before their successor (in forward order), and hence are dominated by the successor.
    Both of these stages run in linear time, and hence we get the following theorem (see \Cref{a:forbiddenwalks} for the pseudocode).
    
    \begin{theorem}[restate = minforbiddenwalks, name = ]
    Given a graph $G := (V, E)$ with $n$ nodes, $m$ edges and $s,t\in V$, the set of minimal \forbiddenwalks{} for \EWmodel{} can be computed in $O(m + n)$ time.
    \end{theorem}
	
	\subparagraph{Compact representation of the solution.}
	
	\begin{figure}[t!]
	    \centering
	    \includegraphics[scale=0.9]{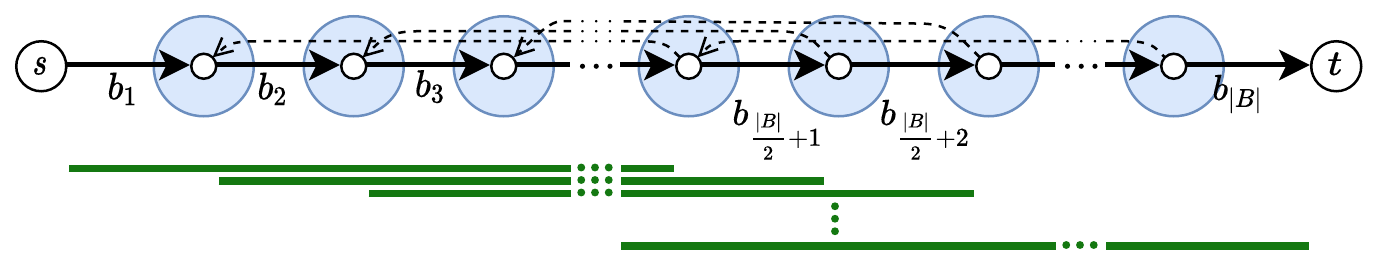}
	    \caption{A graph with maximal safe sequences for \EWmodel{} of total length $\Theta(n^2)$ ($|B|$ is even).}
	    \label{f:n2safe}
	\end{figure}
	
	%In contrast to the previous problems, \forbiddenwalks{} might overlap, which causes the solutions to overlap as well.
	In contrast to the previous problems, solutions to \EW{} might overlap.
	This allows the solutions' total length to be quadratic in the number of nodes (see \Cref{f:n2safe}).
	Therefore, instead of reporting the solution directly, the algorithm creates a {\em compact representation} from which the complete solution can be reported in time linear to its total length.
	%This representation consists of the bridge sequence, and the start and end indices in the bridge sequence for each maximal safe walk in the solution.
	This representation consists of the bridge sequence, and the starts and ends of each maximal safe walk in the bridge sequence.
	
	Now, each such maximal safe sequence begins with the start of previous \forbiddenwalk{} (or the first $s$-$t$ bridge) and ends with the end of current \forbiddenwalk{} (or the last \bridge{}) (see \Cref{l:maxforbiddenwalkstructure}).
	Note that the definition of the start and end for zero bridge length \forbiddenwalks{} (\forbiddentrails) perfectly fits this structure of the solution.
	Hence, the indices of the solution for the {\em compact representation} can be computed by simply iterating over all minimal  \forbiddenwalks{}, requiring $O(|\BrgS|)$ time. % (see \Cref{apn:ewsolution} for more details). %\\ %description and proof of correctness of this algorithm see \Cref{apn:ewsolution}.
	
    Resulting, the minimal \forbiddenwalks{} and the compact representation of the solution of size $O(|B|)$ can be computed in $O(m+n)$ time, which can be expanded to get the complete solution in time linear in the size of the solution.
	Thus, we have proven \Cref{t:lineardatastructure} for $X = V \cup E$, $X = V$ and $X = E$.
	
	%With this algorithm and the algorithms required to compute the minimal \forbiddenwalks{}, we can compute the representation in $O(m+n)$ time.
	%Since in each bridge at most one minimal \forbiddenwalk{} starts, there are at most $O(|B|)$ different solutions, and the compact representation has size $O(|B|)$.
	
	\section{Subset Visibility}
    \label{s:subset-visibility}

    In this section, we discuss $X$-visible variants of our problems. We prove that \XP{}  and \XT{} are in fact NP-hard (\Cref{t:np-hardness}). We then show how  
    %We first prove \Cref{t:np-hardness} and then show how 
    to solve \XW{} as an extension of \EW{}.
    
	\subparagraph{NP-hardness.} 
	We prove \XP{}  and \XT{} to be NP-hard by reduction from the following problem, proven NP-complete in~\cite[Theorem 2]{Fortune1980}.
	
%	The subset visible variants of the \EPmodel{} a%nd \ETmodel{} models are NP-hard.
%	We prove this by reduction from the following problem, proven NP-complete in~\cite[Theorem 2]{Fortune1980}.
	
	\begin{problem}[\textsc{Detour}]
		\label{p:detour}
		Given a graph $G$ and pairwise distinct nodes $u, v, w$ of $G$, decide if there is a $u$-$v$ path in $G$ that contains $w$.
	\end{problem}
	
% 	The reduction works by {\em expansion} of each node in $G$.
% 	In the resulting graph $G'$ each node has either at most one incoming or at most one outgoing edge.
% 	Then, each walk that visits a node twice also visits an edge twice, which makes trails in $G'$ equivalent to paths in $G$.
% 	We set $X = \{u, v, w\}$ and ask if $(u, v)$ is safe in case of $X \subseteq V$, and do the same with their internal edges in case of $X \subseteq E$, proving the result. \hint{remove till here.}
    Observe that with the following reduction, a certificate for the unsafety of a sequence is also a certificate for a detour.
	Therefore, as \detour{} is in NP, our NP-hard problems are in co-NP. 
	Formally, we describe it as follows.

	\nphardness*
	\begin{proof}%[Proof of \Cref{t:np-hardness}]
	Let $G,u,v,w$ be an instance of \textsc{Detour}. In order to address the \ETmodel{} problem in the same way as the \EPmodel{} problem, we transform $G$ into the graph $G'$ by expanding each node.
	For a node $x$ we denote its internal edge as $e_x$.
	
	For the rest of the proof we set $s = \head{}(e_u)$ and $t = \tail{}(e_v)$. Every node of $G'$ has either exactly one incoming edge or exactly one outgoing edge (i.e., some internal edge $e_x$). As such, any $s$-$t$ walk visiting a node twice also visits some edge twice (i.e., some internal edge $e_x$ incident to this repeated node of $G'$). Thus, all $s$-$t$ trails of $G'$ are $s$-$t$ paths.
	
	When we restrict $X \subseteq E$, we set $X = \{e_u,e_w,e_v\}$. We have that $(e_u,e_v)$ is safe under the $X$-visible $s$-$t$ paths model in $G'$ if and only if $G,u,v,w$ is a no-instance for \detour. Since all $s$-$t$ trails of $G'$ are $s$-$t$ paths, the same holds also for the $s$-$t$ trails model.
	
    When we restrict $X \subseteq V$, we set $X = \{s,\tail{}(e_w),t\}$, and analogously have that $(s,t)$ is safe under the $S$-visible $s$-$t$ paths model in $G'$, or under the $s$-$t$ trails model in $G'$, respectively, if and only if $G,u,v,w$ is a no-instance for \detour.
	\end{proof}

	When considering walks, both nodes and edges can be used without limits, and therefore the reduction from \textsc{Detour} does not work.
	Instead, our algorithm for \EW{} can be extended to solve \XW{}.
	
	\subparagraph{Characterisation for \textnormal{\XWmodel{}}.} In \XW{} an additional subset $X \subseteq V \cup E$ of visible nodes and edges  is given.
	By {\em expansion} of each visible node $v$ and making its internal edge visible instead of $v$, we reduce the problem to $X \subseteq E$.
	Similar to previous problems, the solutions to \XW{} are then substrings of the \emph{\xbridgesequence{}}, which is the $X$-subsequence of the \bridgesequence{}.
	%And similar to the \EWmodel{} model, s
	Such a substring is safe if it contains no \emph{\xforbiddenwalk{}}, which is a \forbiddenwalk{} with an $X$-edge.
	With these definitions we get the following characterisation.
	
	\begin{theorem}[restate = charfw, name = Safety for \XWmodel{}]
	    \label{t:charfw}
		A substring of the \xbridgesequence{} is safe under the \XWmodel{} model, if and only if it has no \xforbiddenwalk{}.
    \end{theorem}
    \begin{proof}
% 		We proceed as in the proof of \Cref{t:charew}.
	
		$(\Rightarrow)$ Let $L$ be a substring of the \xbridgesequence{} that is safe under the \XWmodel{} model where $X \subseteq E$.
		If $L$ has an \xforbiddenwalk{}, then from an $s$-$t$ walk $W$ an $s$-$t$ walk $W'$ can be constructed by inserting that \xforbiddenwalk{} into every occurrence of $L$ in the $F$-subsequence of $W$.
		But then $L$ is not a substring of the $F$-subsequence of $W'$.
		
		$(\Leftarrow)$ Let $L := (e_1, \dots, e_{|L|})$ be a substring of the \xbridgesequence{} that has no \xforbiddenwalk{}.
		Let $W$ be an $s$-$t$ walk and let $W_X$ be its $X$-subsequence.
		Since $L$ is a substring of the \xbridgesequence{}, $W_X$ contains a substring $W'_X$ that starts from the last occurrence of $e_1$ and ends in the first occurrence of $e_{|L|}$ after that, and $L$ is a subsequence of $W'_X$.
		We prove that $L$ is a prefix of $W'_X$ by induction.
		By definition, $W'_X$ starts with $e_1$.
		For the inductive step, suppose that $e_1, \dots, e_i$ is a prefix of $W'_X$.
		By definition of $W'$, none of its $\head{}(e_i)$-$\tail{}(e_{i+1})$-subwalks contain $e_1$ or $e_{|L|}$.
		Therefore, since $L$ does not have any \fforbiddenwalk{}, no $e_i$-$e_{i+1}$-subwalk of $W'$ contains an $X$-edge.
		Thus $e_i$ is immediately followed by $e_{i+1}$ in $W'_X$.
	\end{proof}
	
	%\subparagraph{Solving \textnormal{\XW{}}.} 
	From this we can derive an algorithm similar to that for \EW{}.
	After computing the \bridgesequence{} and components, we remove those \bridges{} that are not in $X$ and merge the corresponding \islands{}.
	Then \xforbiddenwalks{} of non-zero bridge length can be computed as before, since their \bridge{} is their $X$-edge.
	\xforbiddenwalks{} of bridge length zero can be found by computing the reverse reachability from $t$ as in \Cref{s:min-forb-walks}, and then iterating over each $X$-edge of a \island{} and checking whether its head is marked.
	The rest of the algorithm remains unchanged.
	This algorithm runs in the same time constraints as that for \EW{}.
	
	\section{Extension to Multigraphs}
	\label{s:extensions}
	
    Most of the results in this paper can be applied to graphs with multiedges without much change.
	The \bridge{} algorithm naturally extends to multiedges.
	Since including multiedges is a generalisation, the NP-hardness results remain valid.
	Even for the $G$-visible problems and the \XWmodel{} problem no change is required.
	In \VPmodel{} and \VWmodel{}, where the multiplicity of edges is not relevant, the parallel edges can simply be merged.
	Only in \VTmodel{} extending to multigraphs (denoted by \emph{\multiVTmodel{}}) is non-trivial, as merging parallel edges changes the set of candidate solutions.
	See for example \Cref{f:multigraphs}, where adding a multiedge in a safe sequence of nodes creates a \forbiddenstructure{}.
	
    \begin{figure}[tbh]
	    \centering
	    \includegraphics[width=.7\textwidth]{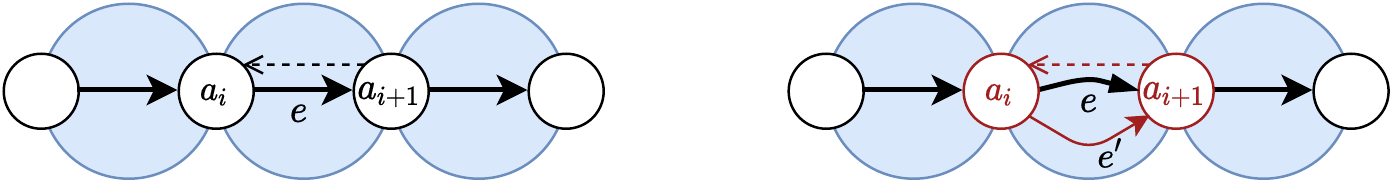}
	    \caption{Example for creating a \forbiddentrail{} by adding a multiedge. Originally, the dashed edge cannot be used as a \forbiddentrail{} without repeating the edge $e$. On adding the parallel edge $e'$, the red cycle becomes a \forbiddentrail{} repeating $a_i$ and $a_{i+1}$ without repeating an edge.
	    The blue regions mark \articulationislands{} and each node is an \articulationpoint{}.}
	    \label{f:multigraphs} 
	\end{figure}
	
	To describe how to solve \textsc{MaxSafe} \multiVTmodel{}, we now assume that $G' := (V, E)$ is a multigraph, $G$ is $G'$ with all parallel edges merged and $s, t \in V$ are given as before.
	First of all, observe that for a sequence of nodes to be safe in \multiVTmodel{}, it needs to be a substring of the \articulationsequence{}.
	Furthermore, we have a trivial necessary condition for safety, which is similar to the adjacency condition in the $G$-visible cases.
	A sequence of two nodes can only be safe if there is no path with more than one edge that connects the first to the second.
	This is equivalent to requiring that the two nodes are connected by an \bridge{} in $G$.
	From here on, we consider a sequence of nodes $L$ that fulfils both of these conditions, where the second is fulfilled by each consecutive pair of nodes.
	
	Let $P$ be the path spelled by $L$, such that the edges of $P$ are a substring of the \bridgesequence{}.
	Such a $P$ is safe in \EPmodel{} in $G$, and hence we can state that: for $L$ to be safe also in \multiVTmodel{} in $G'$, it needs to spell out a path in $G$ that is safe in \EPmodel{}.
	But we can make even more detailed relations to the $G$-visible cases.
	
	\newcommand{\T}{\ensuremath{\mathcal{T}}}
	\newcommand{\W}{\ensuremath{\mathcal{W}}}
	
	For that, we denote the set of $s$-$t$ trails in $G$ as $\T$, the set of $s$-$t$ trails in $G'$ as $\T'$, and the set of $s$-$t$ walks in $G$ as $\W$.
	We furthermore denote with $V(\T)$, $V(\T')$ and $V(\W)$ the sets of node sequences associated with these sets of walks, respectively.
	That is, $V(\T')$ is the candidate set of \multiVTmodel{} in $G'$ and $V(\T)$ and $V(\W)$ are the candidate sets of \VTmodel{} and \VWmodel{} in $G$.
	Observe that we have the following relations: 
	\begin{align}
	V(\T) \subseteq V(\T') \subseteq V(\W)    
	\label{eq:multi-3-inclusions}
	\end{align}
	
	Therefore, if $L$ is safe for $V(\T')$, each element of $V(\T')$ contains $L$, and therefore each element of $V(\T)$ contains $L$ as well, which makes $L$ safe for $V(\T)$.
	The same argument makes $L$ safe for $V(\T')$ if it is safe for $V(\W)$.
	With these relations, we can detail our statement from above by stating that:
	\begin{enumerate}[(a)]
	    \item For $L$ to be safe in \multiVTmodel{} in $G'$, it needs to be safe in \VTmodel{} in $G$ (because of the first inclusion in \Cref{eq:multi-3-inclusions}). This is equivalent to $L$ spelling out a path in $G$ that is safe in \ETmodel{}.
	    \item If $L$ is safe in \VWmodel{} in $G$ (which is equivalent to $L$ spelling out a path in $G$ that is safe in \EWmodel{}), then $L$ is safe in \multiVTmodel{} in $G'$ (because of the second inclusion in \Cref{eq:multi-3-inclusions}).
	\end{enumerate}
	
	So, to compute the solutions of \textsc{MaxSafe} \multiVTmodel{} that are not single nodes, we can start from the \bridgesequence{} in $G$ and proceed similarly to the $G$-visible cases.
	From statement (A), we know that \forbiddentrails{} in $G$ are breaking, while no other structure than a \forbiddenwalk{} in $G$ can be breaking because of statement (B).
	Therefore, the question that remains is: what \forbiddenwalks{} in $G$ of non-zero bridge length are actually breaking?
    To answer this question, observe that a \forbiddenwalk{} in $G$ that contains an edge that is an \bridge{} in $G'$ cannot be used by an $s$-$t$ trail in $G'$ without repeating that edge.
    All other \forbiddenwalks{} can be defined as follows.
    
    \begin{definition}[restate = trailmultibreaker, name = Trail Multi-Breaker]
        Let $G'$ be a multigraph and $G$ be the graph obtained by merging the parallel edges of $G'$.
        We say that $Q$ is a \emph{trail multi-breaker} in $G'$ if $Q$ is a \forbiddenwalk{} in $G$ that contains no \bridge{} of $G'$.
    \end{definition}
    
    We can prove that these are exactly the \forbiddenstructure{}s in \multiVTmodel{}.
    Consider a trail multi-breaker $Q$ that is a \forbiddenwalk{} of non-zero bridge length for $P$.
    Let $P'$ be the set of all parallel edges of $P$ in $G'$ (in addition to the edges of $P$).
	Let $R'$ be an $s$-$t$ path in $G'$.
	It holds that $Q$ and $R'$ can only share edges in $P'$, since otherwise $Q$ would not be a \forbiddenwalk{} for $P$.
	And since $Q$ contains no edge that is an \bridge{} in $G'$, all shared edges are not \bridges{}.
	Therefore, since merging all parallel edges in $P'$ produces the edges of $P$, which are \bridges{}, all edges that $Q$ and $R'$ can share have a parallel edge.
	Replacing the shared edges with the parallel edges in $R'$ produces an $s$-$t$ path in $G'$ in which $Q$ can be inserted.
	Therefore, $Q$ can be used by an $s$-$t$ trail in $G'$.
	Resulting, we get the following theorem.
    
    \begin{theorem}[restate = charmvt, name = ]
        A substring $L$ of the \articulationsequence{} is safe under the \multiVTmodel{} model if and only if it has no trail multi-breaker.
    \end{theorem}
	
	Hence, \textsc{MaxSafe} \multiVTmodel{} can be solved by computing \forbiddenwalks{} in $G$, and filtering to keep only those whose bridge sequence contains no edge that is an \bridge{} in $G'$.
	The solution can then be reported as for \EW{}, and the whole algorithm runs in the same time constraints as \EW{}.

	\section{Conclusions}
	\label{s:conclusions}
	
	On the theoretical side, we considered a natural generalisation of $s$-$t$ bridges, with the notion of safety. We considered the standard solution sets of $s$-$t$ paths, trails and walks, and natural extensions thereof. We fully characterised the complexity of all problems, obtaining a clear trichotomy between linearly solvable problems, problems that allow to compute a compact representation of the solution in linear time, and NP-hard problems.
	
	On the practical side, our problems have potential applications in the genome assembly problem. Observe that, in a sense, our solution sets generalise the set $\mathcal{W}_0$ of circular edge-covering walks, as follows. Take any edge $e$, and consider the set $\mathcal{W}_e$ of all walks from the head of $e$ to the tail of $e$. It holds that any safe walk w.r.t.~$\mathcal{W}_e$ is also safe w.r.t.~$\mathcal{W}_0$.
	
	But in contrast to circular models, our computational formulations can also be applied to non-circular genomes.
	Moreover, they can be applied to scenarios where more than one genome string (i.e.~more chromosomes) has to assembled from a single genome graph, such as when sequencing and assembling a human genome. 
    Furthermore, since some parts of the graph may correspond to errors from the genome sequencing process, not all edges should be covered (i.e., explained) by the genome assembly solution, which motivates removing the edge-covering assumption.
    Moreover, uncertain or complex parts of the graph can be handled by the subset visibility models.
    For example, suspected errors can also be marked as invisible, hiding them in the solutions. Moreover, diploid genomes such as human contain a maternal and a paternal copy of the chromosomes.
    A position where the two copies differ creates a branch in the graph, and this might undesirably break some safe solutions.
    This motivates marking such areas as invisible, to still determine if their flanking regions are consecutive in all solutions to this model.
    As such, one can potentially obtain long safe sequences skipping over invisible parts.
    Finally, all algorithms given here are much simpler than the ones of \cite{DBLP:journals/talg/CairoMART19,Cairo2020pre} (especially when using the simplified $s$-$t$ bridge algorithm~\cite{CairoKRSTZ20}), and thus potentially more suitable for practical applications.% using the simplified $s$-$t$ bridge algorithm~\cite{CairoKRSTZ20}.
	
	%%
	%% Bibliography
	%%
	
	%% Please use bibtex, 
	
	\bibliography{references}
	
	\newpage
	\appendix
	
\end{document}